%% file: poa.tex
\documentclass[11pt]{article}
\usepackage[letterpaper,margin=1in]{geometry}
\usepackage{amsmath,amsthm,amssymb}
\usepackage{graphicx}
\usepackage[ruled,vlined]{algorithm2e}
\usepackage{paralist}
\usepackage{enumitem}           
\usepackage{url}

\newcommand \q {\mathbf{q}}
\renewcommand \b {\mathbf{b}}

\newcommand \x {\mathbf{x}}
\newcommand \y {\mathbf{y}}
\newcommand \z {\mathbf{z}}

\renewcommand \v {\mathbf{v}}

\newcommand \B {\mathbf{B}}
\newcommand \C {\mathbf{C}}

\renewcommand \o {\mathbf{o}}
\renewcommand \t {\mathbf{t}}
\newcommand \w {\mathbf{w}}

\newcommand \one {\mathbf{1}}
\renewcommand \c {\mathbf{c}}
\newcommand \D {\mathbf{D}}
\newcommand \SW {\mathrm{SW}}
\newcommand \EW {\mathrm{EW}}
\DeclareMathOperator*{\E}{\mathbb{E}}

\theoremstyle{definition}
\newtheorem{theorem}{Theorem}

\newtheorem{lemma}[theorem]{Lemma}
\newtheorem{proposition}[theorem]{Proposition}
\newtheorem{definition}[theorem]{Definition}

\theoremstyle{remark}

\newtheorem{remark}[theorem]{Remark}

\title {On the Efficiency of the Proportional Allocation Mechanism for
  Divisible Resources}
\author{George Christodoulou\thanks{University of Liverpool,
    UK. Email:\texttt{\{gchristo,salkmini,bo.tang\}@liverpool.ac.uk} }
		\thanks{This author was supported by EPSRC grants EP/M008118/1 and EP/K01000X/1.}
		\and Alkmini Sgouritsa\footnotemark[1]\and Bo Tang\footnotemark[1]}
\date{}

\begin{document}

\pagestyle{plain}
\maketitle

\input{abstract}
\input{introduction}
\input{preliminary}
\input{concave}
\input{subadditive}

\input{budget}
\input{poly}

\bibliographystyle{plain}
\bibliography{poa}

\end{document}

%% file: abstract.tex
\begin{abstract}
  We study the efficiency of the \emph{proportional allocation
    mechanism}, that is widely used to allocate divisible 
  resources. Each agent submits a bid for each divisible resource and
  receives a fraction proportional to her bids. We quantify the
  inefficiency of Nash equilibria by studying the Price of Anarchy
  (PoA) of the induced game under complete and incomplete
  information. When agents' valuations are concave, we show that the
  Bayesian Nash equilibria can be arbitrarily inefficient, in contrast
  to the well-known $4/3$ bound for pure
  equilibria~\cite{johari_efficiency_2004}. Next, we upper bound the PoA over
  Bayesian equilibria by $2$ when agents' valuations are subadditive,
  generalizing and strengthening previous bounds on lattice submodular
  valuations. Furthermore, we show that this bound is tight and cannot
  be improved by any {\em simple} or {\em scale-free} mechanism. Then we switch to settings with
  budget constraints, and we show an improved upper bound on the PoA over
  coarse-correlated equilibria. Finally, we prove that the PoA is
  {\em exactly} $2$ for pure equilibria in the polyhedral environment.
\end{abstract}

%%% Local Variables: 
%%% mode: latex
%%% TeX-master: "poa"
%%% End: 

%% file: introduction.tex
\section{Introduction}
Allocating network resources, like bandwidth, among agents is a canonical
problem in the network optimization literature. A traditional model
for this problem was proposed by Kelly~\cite{kelly_charging_1997},
where allocating these infinitely divisible resources is treated as a
market with prices. More precisely, agents in the system submit bids
on resources to express their willingness to pay. After soliciting the
bids, the system manager prices each resource with an amount equal to
the sum of bids on it. Then the agents buy portions of resources
proportional to their bids by paying the corresponding prices. This
mechanism is known as the \emph{proportional allocation mechanism} or
Kelly's mechanism in the literature.

The proportional allocation mechanism is widely used in network
pricing and has been implemented for allocating computing resources in
several distributed systems \cite{CC00}. In practice, each agent has
different interests for different subsets and fractions of the
resources. This can be expressed via a {\em valuation} function of the
resource allocation vector, that is typically private knowledge to
each agent. Thus, agents may bid strategically to maximize their own
utilities, i.e., the difference between their valuations and
payments. Johari and Tsitsiklis~\cite{johari_efficiency_2004} observed
that this strategic bidding in the proportional allocation mechanism
leads to inefficient allocations, that do not maximize social
welfare. On the other hand, they showed that this efficiency loss is
bounded when agents' valuations are concave. More specifically, they
proved that the proportional allocation game admits a {\em unique
  pure} equilibrium with Price of Anarchy (PoA)~\cite{KP99} at most
$4/3$.

An essential assumption used by Johari and Tsitsiklis~\cite{johari_efficiency_2004} is that agents
have complete information of each other's valuations. However, in many
realistic scenarios, the agents are only partially informed. A
standard way to model incomplete information is by using the Bayesian
framework, where the agents' valuations are drawn independently from
some publicly known distribution, that in a sense, represents the
agents' beliefs. A natural question is whether the efficiency loss is
still bounded in the Bayesian setting. We give a negative answer to
this question by showing that the PoA over Bayesian equilibria is at
least $\sqrt{m}/2$, where $m$ is the number of resources. This result
complements the current study by Caragiannis and
Voudouris~\cite{CV14}, where the PoA of single-resource proportional
allocation games is shown to be at most $2$ in the Bayesian setting.

Non-concave valuation functions were studied by Syrgkanis and
Tardos~\cite{ST13} for both complete and incomplete information games. They
showed that, when agents' valuations are
lattice-submodular, the PoA for coarse correlated
and Bayesian Nash equilibria is at most $3.73$, by applying their general 
smoothness framework. In this paper, we study subadditive
valuations~\cite{evans1984test} that is a superclass of lattice
submodular functions. We prove that the PoA over Bayesian Nash equilibria
is at most $2$. Moreover, we show optimality of the proportional
allocation mechanism, by showing that this bound is tight and
cannot be improved by any \emph{simple} mechanism
, as defined in the recent framework of
Roughgarden~\cite{Rou14}\footnote{In a simple mechanism, the agents' action space should be 
at most sub-doubly-exponential in
 $m$.}, or any {\em scale-free} mechanism\footnote{The 
basic property of a scale-free mechanism is that, if every bid is scaled 
by the same constant, the outcome remains unchanged (we refer the reader to Section 
\ref{sec:scalefree} for the complete definition).}. 

Next, we switch to the setting where agents are constrained
by budgets, that represent the maximum payment they can afford. We
prove that the PoA of the proportional allocation mechanism is at most
$1+\phi\approx 2.618$, where $\phi$ is the golden ratio. The previously
best known bound was $2.78$ and for a single resource due
to~\cite{CV14}. 
Finally, we consider the polyhedral environment that was previously
studied by Nguyen and Tardos in \cite{nguyen_approximately_2007},
where they proved that pure equilibria are at least $75\%$ efficient
with concave valuations. We prove that the PoA is exactly $2$ for agents
with subadditive valuations.

{\bf Related Work.}
The efficiency of the proportional allocation mechanism has been extensively
studied in the literature of network resource allocation. Besides the
work mentioned above, Johari and Tsitsiklis \cite{JT09} studied a more
general class of scale-free mechanisms and proved that the proportional 
allocation mechanism achieves the best PoA in this class.
Zhang~\cite{Zhang05} and Feldman et al.~\cite{Feldman2005} studied the
efficiency and fairness of the proportional allocation mechanism, when
agents aim at maximizing non quasi-linear utilities subject to budget constraints. 
Correa, Schulz and Stier-Moses \cite{CSS13}
showed a relationship in the efficiency loss between proportional
allocation mechanism and non-atomic selfish routing for not necessarily
concave valuation functions.

There is a line of research studying the PoA of simple
auctions for selling indivisible goods~(see
\cite{BR11,CKS08,HKMN11,ST13}%for an incomplete list
). Recently,
Feldman et al. \cite{feldman_simultaneous_2012} showed tighter upper
bounds for simultaneous first and second price auctions when the
agents have subadditive valuations. Christodoulou et al.~\cite{CKST13}
showed matching lower bounds for simultaneous first price auctions, and
Roughgarden~\cite{Rou14} proved general lower bounds for the PoA of
all simple auctions, by using the corresponding computational or
communication lower bounds of the underlying allocation problem.

%%% Local Variables: 
%%% mode: latex
%%% TeX-master: "poa"
%%% End: 

%% file: preliminary.tex
\section{Preliminaries}
\label{sec:preliminary}

There are $n$ {\em agents} who compete for $m$ {\em divisible
  resources} with {\em unit} supply. Every agent $i\in [n]$ has a
valuation function $v_i: [0,1]^m \rightarrow \mathbb R_+$, where $[n]$
denotes the set $\{1,2,\dots,n\}$. The valuations are normalized as
$v_i(\mathbf{0})=0$, and monotonically non-decreasing, that is,
for every $\x,\x' \in [0,1]^m$, where $\x=(x_j)_j, \x'=(x'_j)_j$ and
$\forall j\in[m] \; x_j\le x'_j $, we have $v_i(\x)\leq v_i(\x')$. Let $\x+\y$
be the componentwise sum of two vectors $\x$ and $\y$.

\begin{definition}
  A function $v: [0,1]^m\rightarrow \mathbb{R}_{\ge 0}$ is
  subadditive if, for all $\x,\y\in [0,1]^m$, such
  that $\x+\y \in[0,1]^m$, it is $v(\x + \y)\leq v(\x)+v(\y)$.
\end{definition}
\begin{remark}
Lattice submodular functions used in \cite{ST13} are subadditive (see
Section \ref{sec:Proportional}).
In the case of a single variable
(single resource), any concave function is subadditive; more
precisely, concave functions are equivalent to lattice submodular
functions in this case.  However, concave functions of many variables
may not be subadditive~\cite{Ros50}.
\end{remark}

In the \emph{Bayesian} setting, the valuation of each agent $i$ is
drawn from a set of possible valuations $V_i$, according to some known
probability distribution $D_i$. We assume that $D_i$'s are independent, but not
necessarily identical over the agents. 

A mechanism can be represented by a tuple $(\x,\q)$, where $\x$
specifies the allocation of resources and $\q$ specifies the agents'
payments. In the mechanism, every agent $i$ submits a non-negative bid
$b_{ij}$ for each resource $j$.  
The proportional allocation mechanism determines the allocation
$x_i=(x_{ij})_j$ 
and payment $q_i$, for each agent $i$, as follows:
$x_{ij}=\frac{b_{ij}}{\sum_{k\in[n]}b_{kj}}$,
$q_i=\sum_{j\in [m]}b_{ij}.$ When all agents bid $0$, the allocation
can be defined arbitrarily, but consistently.

{\bf Nash Equilibrium.} 
We denote by $\b=(b_1,\ldots,b_n)$ the strategy profile of all agents, 
where $b_i=(b_{i1},\ldots,b_{im})$ denotes the pure bids of agent $i$ 
for the $m$ resources. 
By $\b_{-i}=(b_1,\ldots,b_{i-1},b_{i+1},\ldots,b_n)$ we denote the
strategies of all agents except for $i.$ Any \emph{mixed, correlated, coarse
  correlated or Bayesian strategy $B_i$} of agent $i$ is a probability
distribution over $b_i$. For any strategy profile $\b$, $\x(\b)$
denotes the allocation and $\q(\b)$ the payments under the strategy
profile $\b$. The \emph{utility $u_i$} of agent $i$ is defined as the
difference between her valuation for the received allocation and her
payment: $u_i(\x(\b),\q(\b))=u_i(\b)=v_i(x_i(\b))-q_i(\b)$.

\begin{definition}
A bidding profile $\B$ forms the following equilibrium if for every agent $i$ and all bids $b'_i$: 

\noindent Pure Nash equilibrium:
\emph{$\B=\b$, $u_i(\b)\ge u_i(b'_i, \b_{-i})$}.

\noindent Mixed Nash equilibrium:
\emph{$\B=\times_i B_i$, $\E_{\b\sim\B}[u_i(\b)]\ge \E_{\b\sim
  \B}[u_i(b'_i,\b_{-i})]$}.

\noindent Correlated equilibrium:
\emph{$\B=\left(B_i\right)_i$, $\E_{\b\sim\B}[u_i(\b)|b_i]\ge \E_{\b\sim\B}[u_i(b'_i,\b_{-i})|b_i]$}.

\noindent Coarse correlated equilibrium:
\emph{$\B=\left(B_i\right)_i$, $\E_{\b\sim\B}[u_i(\b)]\ge \E_{\b\sim \B}[u_i(b'_i,\b_{-i})]$}.

\noindent Bayesian Nash equilibrium:
\emph{$\B(\v)=\times_iB_i(v_i)$, $\E_{\v_{-i},\b}[u_i(\b)]\ge \E_{\v_{-i},\b}[u_i(b'_i, \b_{-i})]$}.

\end{definition}
The first four classes of equilibria are in increasing order of
inclusion. Moreover, any mixed Nash equilibrium is also a Bayesian
Nash equilibrium.

{\bf Price of Anarchy (PoA).}
Our global objective is to maximize the sum of the agents' valuations
for their received allocations, i.e., to maximize the \emph{social
  welfare} $\SW(\x)=\sum_{i\in [n]} v_i(x_i).$ Given the valuations,
$\v$, of all agents, there exists an optimal allocation
$\o^{\v}=\o=(o_1,\ldots,o_n)$, such that $\SW(\o)=\max_{\x} \SW(\x)$.
By $o_i = (o_{i1}, \ldots, o_{im})$ we denote the optimal allocation
to agent $i$. For simplicity, we use $\SW(\b)$ and $v_i(\b)$ instead of
$\SW(\x(\b))$ and $v_i(x_i(\b))$, whenever the allocation rule $\x$ is
clear from the context. We also use shorter notation for expectations,
e.g. we use $\E_{\v}$ instead of $\E_{\v\sim \mathbf{D}}$, 
$\E[u_i(\b)]$ instead of $\E_{\b\sim \B}[u_i(\b)]$ and
$u(\B)$ for $\E_{\b \sim \B}[u(\b)]$ whenever $\D$ and $\B$ are clear
from the context.

\begin{definition} Let $\mathcal{I}([n],[m],\v)$ be the set of all
instances, i.e., $\mathcal{I}([n],[m],\v)$ includes
the instances for every set of agents and resources and any possible
valuations that the agents might have for the resources. We define the
pure, mixed, correlated, coarse correlated and Bayesian Price of
Anarchy, PoA, as
$$\text{PoA} = \max_{I \in \mathcal{I}} \max_{\B \in \mathcal{E}(I)}
\frac{\E_{\substack{\v}}[\SW(\o)]}{\E_{\substack{\v,
      \b\sim\B}}[\SW(\b)]}, $$ where $\mathcal{E}(I)$ is the set of
pure Nash, mixed Nash, correlated, coarse correlated or Bayesian Nash
equilibria for the specific instance $I \in \mathcal{I}$, 
respectively\footnote{The expectation over $\v$ is only
needed for the definition of Bayesian PoA.}. 
\end{definition}

{\bf Budget Constraints.} 
We also consider the setting where agents are budget-constrained.
That is, the payment of each agent $i$ cannot be higher than $c_i$,
where $c_i$ is a non-negative value denoting agent $i$'s budget.
Following \cite{CV14,ST13}, we use {\em Effective Welfare} as the
benchmark: $\EW(\x) = \sum_i \min\{v_i(x_i),c_i\}$. In addition, for
any {\em randomized} allocation $\x$, the expected effective welfare
is defined as:
$\E_{\x}[\EW(\x)] = \sum_i \min\{\E_{\x}[v_i(x_i)],c_i\}.$

%%% Local Variables: 
%%% mode: latex
%%% TeX-master: "poa"
%%% End: 

%% file: concave.tex
\section{Concave Valuations}
\label{sec:concave}

In this section, we show that for concave valuations on multiple
resources, Bayesian equilibria can be arbitrarily inefficient. More
precisely, we prove that the Bayesian PoA is $\Omega(\sqrt{m})$ in
contrast to the constant bound for pure
equilibria~\cite{johari_efficiency_2004}.  Therefore, there is a big
gap between complete and incomplete information settings. We state our
main theorem in this section as follows.

\begin{theorem}
  When valuations are concave, the PoA of the proportional allocation
  mechanism for Bayesian equilibria is at least
  $\frac{\sqrt{m}}2$.
\end{theorem}

\begin{proof}
  We consider an instance with $m$ resources and $2$ agents with the
  following concave valuations. $v_1(\x)=\min_j\{x_j\}$
  and $v_2(\x)$ is drawn from a distribution $D_2$, such that some resource $j\in [m]$ 
	is chosen uniformly at random 
  and then $v_2(\x)=x_j/\sqrt{m}$. Let
  $\delta=1/(\sqrt{m}+1)^2$.  We claim that
  $\b(\v) = (b_1,b_2(v_2))$ is a pure Bayesian Nash equilibrium, where
  $\forall j\in[m]$, $b_{1j}=\sqrt{\delta / m} - \delta$ and, if $j\in [m]$ is 
	the resource chosen by $D_2$, $b_{2j}(v_2)=\delta$ and for all $j'\neq j$ $b_{2j'}=0$.

  Under this bidding profile, agent $1$ bids the same value for all
  resources, and agent $2$ only bids positive value for a single
  resource associated with her valuation. Suppose that agent $2$ has 
	positive valuation for resource $j$, i.e.,  
  $v_2(\x)=x_j/\sqrt m$. Then the rest $m-1$ resources are allocated
  to agent $1$ and agents are competing for resource $j$. Bidder
  $2$ has no reason to bid positively for any other resource. If
  she bids any value $b'_{2j}$ for resource $j$, her utility would be
  $u_2(\b_1,b'_{2j}) =
  \frac{1}{\sqrt{m}}\frac{b'_{2j}}{b_{1j}+b'_{2j}} - b'_{2j}$,
  which is maximized for
  $b'_{2j} = \sqrt{\frac{b_{1j}}{\sqrt{m}}}-b_{1j}$. For $b_{1j}=\sqrt{\delta / m} - \delta$, 
	the utility of agent $2$ is maximized for
  $b'_{2j} = 1/(\sqrt{m}+1)^2=\delta$ by
  simple calculations.

Since
  $v_1(\x)$ equals the minimum of $\x$'s components, agent $1$'s valuation
  is completely determined by the allocation of resource $j$. So the
  expected utility of agent $1$ under $\b$ is
  $\E_{v_2}[u_1(\b)] = \frac{\sqrt{\delta / m} - \delta}{\sqrt{\delta / m} - \delta+\delta}-m(\sqrt{\delta / m} - \delta)=
  (1-\sqrt{m\delta})^2 =
  \frac{1}{\left( \sqrt{m}+1\right)^2}=\delta$.
  Suppose now that agent $1$ deviates to $b'_1 = (b'_{11}, \ldots , b'_{1m})$. 
\begin{align*}
   \E_{v_2}[u_1(b'_1,b_2)] &= \frac 1m \sum_j \frac{b'_{1j}}{b'_{1j}+\delta} - \sum_j b'_{1j}
		= \frac 1m \sum_j \left(\frac{b'_{1j}}{b'_{1j}+\delta} - m \cdot b'_{1j}\right)\\
		&\leq \frac 1m \sum_j \left(\frac{\sqrt{\delta/m} - \delta}{\sqrt{\delta/m}} - m \cdot (\sqrt{\delta/m} - \delta)\right)\\
		&= \frac 1m \sum_j \left(1-2 \sqrt{m\cdot \delta} + m\cdot \delta\right)
		= \frac 1m \sum_j \left(1 - \sqrt{m\cdot \delta} \right)^2\\
		& = \frac 1m \sum_j \left(\frac{1}{\sqrt{m}+1} \right)^2 =\delta = \E_{v_2}[u_1(\b)].
  \end{align*}
	The inequality comes from the fact that $\frac{b'_{1j}}{b'_{1j}+\delta} - m \cdot b'_{1j}$ 
	is maximized for $b'_{1j} = \sqrt{\delta/m} - \delta$. So we
        conclude that
  $\b$ is a Bayesian equilibrium.

  Finally we compute the PoA. The expected social welfare under $\b$
  is $\E_{v_2}[\SW(\b)]= \frac {\sqrt{\delta / m} - \delta}{\sqrt{\delta / m} - \delta+\delta} + \frac{1}{\sqrt{m}}\frac
  {\delta}{\sqrt{\delta / m} - \delta+\delta} = 1-\sqrt{m\delta}+\sqrt{\delta} = \frac
  2{\sqrt{m}+1} < \frac{2}{\sqrt{m}}.$
  But the optimal social welfare is $1$ by allocating to agent $1$ all resources.
  So, PoA $\ge\frac{\sqrt{m}}2.$
\end{proof}

%%% Local Variables: 
%%% mode: latex
%%% TeX-master: "poa"
%%% End: 

%% file: subadditive.tex
\section{Subadditive Valuations}
\label{sec:Proportional}

In this section, we focus on agents with subadditive valuations.  
We prove that the proportional allocation mechanism is at least
$50\%$ efficient for coarse correlated equilibria and Bayesian
Nash equilibria, i.e., $\mathrm{PoA} \le 2$. We further show that this bound is
tight and cannot be improved by any simple or scale-free mechanism. 
Before proving our PoA bounds, we show that the class of subadditive functions 
is a superclass of lattice submodular functions.

\begin{proposition}
  Any lattice submodular function $v$ defined on $[0,1]^m$ is
  subadditive.
\end{proposition}
\begin{proof}
  It has been shown in \cite{ST13} that for any lattice submodular
  function $v(x)$, $\frac{\partial^2 v(x)}{(\partial x_j)^2}\le 0$ and
  $\frac{\partial^2 v(x)}{\partial x_j\partial x_{j'}}\le 0$. So the
  function $\frac{\partial v}{\partial x_j}(x)$ is non-increasing
  monotone for each coordinate $x_{j'}$. It suffices to prove that 
  for any $\x, \y\in [0,1]^m$, $v(\x+\y)-v(\y)\le
  v(\x)-v(\mathbf{0})$. Let $\z^k$ be the vector that $z^k_j=y_j$ if
  $j\le k$ and $x_j+y_j$ otherwise. Note that $\z^0=\x+\y$ and
  $\z^m=\y$. Similarly, we define $\w^k$ to be the vector that
  $\w^k_j=0$ if $j\le k$ and $x_j$ otherwise. It is easy to see that
  $\z^k\ge \w^k$ for all $k\in [m]$. So we have,
  \begin{align*}
    &v(\x+\y)-v(\y)=\sum_{j\in [m]}v(\z^{j-1})-v(\z^j)
    =\sum_{j\in[m]}\int_{y_j}^{x_j+y_j}\frac{\partial v}{\partial x_j}(t_j;\z^j_{-j})dt_j\\
    \le&\sum_{j\in[m]}\int_{y_j}^{x_j+y_j}\frac{\partial v}{\partial x_j}(t_j-y_j;\z^j_{-j})dt_j
    \le\sum_{j\in[m]}\int_{0}^{x_j}\frac{\partial v}{\partial
      x_j}(s_j;\w^j_{-j})ds_j= v(\x)-v(\mathbf{0})
  \end{align*}
 The second equality is due to the definition of partial derivative
 and the inequalities is due to the monotonicity of $\frac{\partial v}{\partial x_j}(x)$.
\end{proof}

\subsection{Upper bound}
\label{sec:subadditiveUB} 
A common approach to prove PoA upper bounds is to
find a deviation with proper utility bounds and then use the definition
of Nash equilibrium to bound agents' utilities at equilibrium. The
bidding strategy described in the following lemma is for this purpose.

\begin{lemma}
  \label{lem:proportion}
  Let $\v$ be any subadditive valuation profile and $\B$ be some
  randomized bidding profile. For any agent $i$, there exists a
  randomized bidding strategy $a_i(\v,\B_{-i})$ such that:
  $$\sum_iu_i(a_i(\v,\B_{-i}), \B_{-i})\ge \frac 12
  \sum_iv_i(o_i^{\v})-\sum_i\sum_j\E_{\b \sim \B}[b_{ij}].$$
\end{lemma}
\begin{proof}
  Let $p_{ij}$ be the sum of the bids of all agents except $i$ on
  resource $j$, i.e., $p_{ij}=\sum_{k\neq i}b_{kj}$. Note that $p_{ij}$
  is a random variable that depends on $\b_{-i}\sim
  \B_{-i}$. Let $P_i$ be the propability distribution of $p_i=(p_{ij})_j$. 
	Inspired by \cite{feldman_simultaneous_2012}, we consider
  the bidding strategy $a_i(\v,\B_{-i}) = (o_{ij}^{\v}\cdot b'_{ij})_j$, 
	where $b'_i\sim P_i$. Then, $u_i(a_i(\v,\B_{-i}), \B_{-i})$ is 
	\begin{align*}
  &\E_{b'_i\sim P_i}\E_{p_i\sim
    P_i}\left[v_i\left(\left(\frac{o^{\v}_{ij}b'_{ij}}{o^{\v}_{ij}b'_{ij}+p_{ij}}\right)_j\right)-o^{\v}_i\cdot b'_i\right]\\
    \ge& \frac 12 \cdot \E_{p_i\sim P_i}\E_{b'_i\sim
      P_i}\left[v_i\left(\left(\frac{o^{\v}_{ij}b'_{ij}}{o^{\v}_{ij}b'_{ij}+p_{ij}}+\frac{o^{\v}_{ij}p_{ij}}{o^{\v}_{ij}p_{ij}+b'_{ij}}\right)_j\right)\right]-\E_{p_i\sim P_i}[o^{\v}_i\cdot p_i]\\
    \ge& \frac 12 \cdot \E_{p_i\sim P_i}\E_{b'_i\sim
      P_i}\left[v_i\left(\left(\frac{o^{\v}_{ij}(b'_{ij}+p_{ij})}{b'_{ij}+p_{ij}}\right)_j\right)\right]-\E_{p_i\sim P_i}[o^{\v}_i\cdot p_i]\\
    =&\frac 12 \cdot v_i(o^{\v}_i)-\sum_j \sum_{k\neq i}\E_{\b\sim \B}[o^{\v}_{ij}\cdot b_{kj}]
  \end{align*}
  The first inequality follows by swapping $p_{ij}$ and $b'_{ij}$ and
  using the subadditivity of $v_i$. The second inequality comes from
  the fact that $o^{\v}_{ij}\le 1$. The lemma follows by summing up
  over all agents and the fact that $\sum_{i\in [n]}o_{ij}^{\v}= 1$.
\end{proof}

\begin{theorem}
  \label{thm:pro_mixed}
  The coarse correlated PoA of % games induced by
  the proportional allocation mechanism
  with subadditive agents is at most $2$.
\end{theorem}

\begin{proof}
  Let $\B$ be any coarse correlated equilibrium (note that $\v$ is
  fixed). By Lemma \ref{lem:proportion} and the definition of the
  coarse correlated equilibrium, we have
  \[ \sum_iu_i(\B)\ge \sum_iu_i(a_i(\v,\B_{-i}), \B_{-i})\ge \frac 12
  \sum_iv_i(o_i)-\sum_i\sum_{j}\E[b_{ij}]\]
  By rearranging terms,
  $\SW(\B)=\sum_iu_i(\B)+\sum_i\sum_{j}\E[b_{ij}]\ge \frac 12 \cdot
  \SW(\o)$.
\end{proof}

\begin{theorem}
  \label{thm:pro_bayes}
  The Bayesian PoA of % games induced by
  the proportional allocation mechanism with subadditive agents is at most
  $2$.
\end{theorem}

\begin{proof}

  Let $\B$ be any Bayesian Nash Equilibrium and let $v_{i}\sim D_{i}$ be
  the valuation of each agent $i$ drawn independently from $D_i$.  We
  denote by $\C=(C_1,C_2,\ldots,C_n)$ the bidding distribution in $\B$
  which includes the randomness of both the bidding strategy $\b$ and
  of the valuations $\v$. %, that is $b_i(v_i)\sim C_i$. 
	The utility of
  agent $i$ with valuation $v_i$ can be expressed by $u_i(\B_i(v_i),
  \C_{-i})$. It should be noted that $\C_{-i}$ does {\em not} depend on some particular $\v_{-i}$, 
	but merely on $\D_{-i}$ and $\B_{-i}$. For any agent $i$ and any subadditive
  valuation $v_i\in V_i$, consider the deviation $a_i(v_i;\w_{-i},\C_{-i})$ as defined
  in Lemma \ref{lem:proportion}, where $\w_{-i}\sim \D_{-i}$. By the definition of the Bayesian Nash
  equilibrium, we obtain
  \[\E_{\v_{-i}}[u^{v_i}_i(\B_i(v_i),\B_{-i}(\v_{-i}))]= u^{v_i}_i(\B_i(v_i),\C_{-i}) 
	\ge \E_{\w_{-i}}[u^{v_i}_i(a_i(v_i;\w_{-i},\C_{-i}), \C_{-i})].\]
  By taking expectation over $v_i$ and summing up over all agents, 
  \begin{align*}
    &\sum_i \E_{\v}[u_i(\B(\v))] \ge\sum_i\E_{v_i,\w_{-i}}[u^{v_i}_i(a_i(v_i;\w_{-i},\C_{-i}),\C_{-i})]\\
    =&\E_{\v}\left[\sum_iu_i^{v_i}(a_i(\v,\C_{-i}),
      \C_{-i})\right]\ge\frac 12\cdot\sum_i\E_{\v}[v_i(o_i^{\v})]-\sum_i\sum_j\E[b_{ij}]
  \end{align*}
  So, $\E_{\v}[\SW(\B(\v))]= \sum_i\E_{\v}[u_i(\B(\v))]+ \sum_i\sum_j\E[b_{ij}]\ge
  \frac 12 \cdot \E_{\v}[\SW(\o^{\v})].$
\end{proof}

\subsection{Simple mechanisms lower bound}
\label{sec:simpleLB}
Now, we show a lower bound that applies to
all simple mechanisms, where the bidding space has size (at most)
sub-doubly-exponential in $m$. More specifically, we apply the general
framework of Roughgarden~\cite{Rou14}, for showing lower bounds on the
price of anarchy for {\em all} simple mechanisms, via communication
complexity reductions with respect to the underlying optimization problem. 
In our setting, the problem is to maximize the social welfare by
allocating divisible resources to agents with subadditive valuations. 
We proceed by proving a communication lower bound for this problem in
the following lemma.

\begin{lemma}
  \label{lem:inappr}
  For any constant $\varepsilon>0$, any
  $(2-\varepsilon)$-approximation (non-deterministic) algorithm for
  maximizing social welfare in resource allocation problem with
  subadditive valuations, requires an exponential amount of
  communication.
\end{lemma}

\begin{proof}
  We prove this lemma by reducing the communication lower bound for
  combinatorial auctions with general valuations (Theorem 3 of
  \cite{Nis02}) to our setting (see also \cite{DNS10} for a reduction
  to combinatorial auctions with subadditive agents). %Theorem 4.1 of \cite{DNS10}.

  Nisan~\cite{Nis02} used an instance with $n$ players and $m$ items,
  with $n < m^{1/2-\varepsilon}$. Each player $i$ is associated with a
  set $T_i$, with $|T_i| = t$ for some $t>0$.  At every instance of
  this problem, the players' valuations are determined by sets $I_i$
  of bundles, where $I_i \subseteq T_i$ for every $i$.  Given $I_i$,
  player $i$'s valuation on some subset $S$ of items is $v_i(S)=1$, if
  there exists some $R\in I_i$ such that $R\subseteq S$, otherwise
  $v_i(S)=0$. In \cite{Nis02}, it was shown that distinguishing
  between instances with optimal social welfare of $n$ and $1$,
  requires $t$ bits of communication. By choosing $t$ exponential in
  $m$, their theorem follows.

  We prove the lemma by associating any valuation $v$ of the above
  combinatorial auction problem, to some appropriate subadditive
  valuation $v'$ for our setting.  For any player $i$ and any
  fractional allocation $\x = (x_1,\ldots, x_m)$, let $A_{x_i} = \{j|x_{ij} >
  \frac 12\}$. We define $v'_i(x_i)=v_i(A_{x_i})+1$ if $x_i\neq
  \mathbf{0}$ and $v'_i(x_i)=0$ otherwise. It is easy to verify that
  $v'_i$ is subadditive. 
  Notice that $v'_i(x) =2$ only if there exists $R\in I_i$ such that
  player $i$ is allocated a fraction higher than $1/2$ for every
  resource in $R$. The value $1/2$ is chosen such that no two players
  are assigned more than that fraction from the same resource. This
  corresponds to the constraint of an allocation in the combinatorial auction where no
  item is allocated to two players.

  Therefore, in the divisible goods allocation problem, distinguishing
  between instances where the optimal social welfare is $2n$ and $n+1$
  is equivalent to distinguishing between instances where the optimal
  social welfare is $n$ and $1$ in the corresponding combinatorial
  auction and hence requires exponential, in $m$, number of
  communication bits.
\end{proof}

The PoA lower bound follows the general reduction described in
\cite{Rou14}.

\begin{theorem}
\label{thm:LBsimple}
The PoA of $\epsilon$-mixed Nash equilibria\footnote{A bidding profile $\B=\times_i B_i$ 
is called $\epsilon$-mixed Nash equilibrium if, for every agent $i$ and all bids $b'_i$, 
$\E_{\b\sim\B}[u_i(\b)]\ge \E_{\b\sim\B}[u_i(b'_i,\b_{-i})] - \epsilon$.} of every simple mechanism,  
when agents have subadditive valuations, is at least $2$. 
\end{theorem}

\begin{remark}
  This result holds only for $\epsilon$-mixed Nash
  equilibria. Considering exact Nash equilibria, we 
  show a lower bound for all \emph{scale-free} mechanisms 
  in the following section.
\end{remark}

\subsection{Scale-free mechanisms lower bound}
\label{sec:scalefree}
Here we prove a tight lower bound for all scale-free mechanisms including 
the proportional allocation mechanism. 
A mechanism $(\x,\q)$ is said to be scale-free if a) for every
agent $i$, resource $j$ and constant $c>0$,
$x_{i}(c\cdot \b_j) = x_{i}(\b_j)$.  Moreover, for a fixed
$\b_{-i}$, $x_{i}(\cdot)$ is non-decreasing and positive whenever
$b_{ij}$ is positive. % $b_{ij} > 0$,
% $x_{i}(\b_j) > 0$.
b) The payment for agent $i$ depends only on her bids
$b_i = (b_{ij})_j$ and equals to $\sum_{j\in [m]} q_i(b_{ij})$ where
$q_{i}(\cdot)$ is non-decreasing, continuous, normalized
($q_{i}(0)=0$), and there always exists a bid $b_{ij}$ such that
$q_{i}(b_{ij})>0$.
 
\begin{theorem}
  \label{thm:generalpayLBSub} The mixed PoA of scale-free mechanisms
  when agents have subbaditive valuations, is at least $2$.
\end{theorem}
\begin{proof}
  Given a mechanism $(\x,\q)$, we construct an instance with $2$
  agents and $m$ resources. Let $V$ be a positive value such that
  $V/m$ is in the range of both $q_1$ and $q_2$. This can be always
  done due to our assumptions on $q_i$. Let $T_1$ and $T_2$ be the
  values such that $q_1(T_1) = q_2(T_2)=V/m$. W.l.o.g. we assume that
  $T_1 \geq T_2$. By monotonicity of $q_1$, $q_1(T_2) \leq V/m$. 
  Pick an arbitrary value $a\in (0,1)$, and let $h_1= x_1(a,a)$ and $h_2 =
  x_2(a,a)$. 
  By the assumption that $x_{i}(\b_j) > 0$ for $b_{ij}
  > 0$, we have $h_1, h_2 \in (0,1)$. Let $v= V/\sqrt{m}$. We define
  the agents' valuations as:
\begin{center}
  \begin{tabular}{cc}
$v_1(x) = \left\{
\begin{array}{l l}
0, & \text{if }  \forall j\in[m], x_{j}=0,  \\
v, & \text{if }  \forall j \; x_{j}<h_1, \;  \exists k \; x_{k} > 0 \\
2v, & \textrm{otherwise} \\
\end{array} \right.$ \qquad & \qquad
$v_2(x)= \left\{
\begin{array}{l l}
0, &  \text{if }  \forall j\in[m], \; x_{j}=0 \quad \\
V, &  \text{if }  \exists j\; x_{j} < h_2, \; \exists k \; x_{k} > 0\\
2V, & \text{otherwise}\\
\end{array} \right.$
\end{tabular}
\end{center}

We claim that the following mixed strategy profile $\B$ is a Nash
equilibrium.  Agent $1$ picks resource $l$ uniformly at random and
bids $b_{1l}=y$, and $b_{1k}=0,$ for $k\neq l$, where $y$ is a random
variable drawn by the cumulative distribution $G(y)=\frac{mq_{2}(y)}{V}, \;\;y\in [0,T_2].$
Agent $2$ bids $b_{2j}=z$ for every item $j$, where $z$ is a random
variable drawn from $F(z)$, defined as
$F(z)=\frac{v-q_1(T_2)+q_{1}(z)}{v},\;\;z\in [0,T_2].$
Recall that $v=V/\sqrt{m}$ and $q_1(T_2) \leq V/m$. Therefore,
$v-q_1(T_2) \geq 0$ and thus $F(0) \geq 0$.  Notice that $G(\cdot)$
and $F(\cdot)$ are valid CDFs, due to monotonicity of $q_i(\cdot)$.
Since $G(T_2) = 1$, $F(T_2)=1$ and $q_i(\cdot)$ is continuous, $G(y)$
and $F(y)$ are continuous in $(0,\infty)$ and therefore both functions
have no mass point in any $y\neq 0$. We assume that if both agents bid
$0$ for some resource, agent $2$ takes the whole resource.
We are ready to show that $\B$ is a Nash equilibrium. For the
following arguments notice that $G(T_2) = 1$, $F(T_2)=1$ and
$G(0) = 0$.

If agent $1$ bids any $y$ in the range $(0,T_2]$ for a single resource
$j$ and zero for the rest, then she gets allocation of at least
$h_1$ (that she values for $2v$), only if $y \geq z$, which happens
with probability $F(y)$. This holds due to monotonicity of
$x_1(\cdot)$ with respect to $y$.  Otherwise her value is
$v$. Therefore, her expected valuation is $v + F(y)v$. So, for every
$y \in (0,T_2]$ her expected utility is
$v+F(y)v-q_1(y) = 2v-q_1(T_2)$.  If agent $1$ picks $y$ according to
$G(y)$, her utility is still $2v-q_1(T_2)$, since she bids $0$ with
zero probability. Suppose agent $1$ bids
$\mathbf{y} = (y_1,\ldots , y_m)$, $y_j \in [0,T_2]$ for every $j$,
with at least two positive bids, and w.l.o.g., assume
$y_1 = \max_j y_j$. If $z > y_1$, agent $1$ has value $v$ for the
allocation she receives.  If $z \leq y_1$, agent $1$ has value $2v$,
but she pays more than $q_1(y_1)$. So, this strategy is dominated by
the strategy of bidding $y_1$ for the first resource and zero for the
rest.  Bidding greater than $T_2$ for any resource is dominated by the
strategy of bidding exactly $T_2$ for that resource.

If agent $2$ bids $z \in [0, T_2]$ for all resources, she gets an
allocation of at least $h_2$ for all the $m$ resources with
probability $G(z)$ (due to monotonicity of $x_2(\cdot)$ with respect
to $z$ and to the tie breaking rule). So, her expected utility is
$V + G(z)V - mq_2(z) = V$. Bidding greater than $T_2$ for any resource
is dominated by bidding exactly $T_2$ for this resource. Suppose that
agent $2$ bids any $\mathbf{z} = (z_1,\dots z_m)$, with
$z_j \in [0,T_2]$ for every $j$, then, since agent $1$ bids positively
for any item with probability $1/m$, agent's $2$ expected utility is
$\frac{1}{m}\sum_{j} \left(V + G(z_j)V-\sum_k q_2(z_k) \right)
=\frac{1}{m}\sum_j \left(V + mq_2(z_j)-\sum_{k} q_2(z_k)\right)
=\frac{1}{m}\left( mV+m\sum_{j} q_2(z_j) -m\sum_{k} q_2(z_k) \right) =
V.$ So, $\B$ is Nash equilibrium. 

Therefore, it is sufficient to bound the expected social welfare in
$\B$. Agent $1$ bids $0$ with zero probability. So, whenever 
agent $2$ bids $0$, she receives exactly $m-1$ resources, which 
she values for $V$. Agent $2$ bids $0$ with probability $F(0) = 1-\frac{q_1(T_2)}{v}
\geq 1-\frac{V}{mv} = 1-\frac{1}{\sqrt{m}}$. Hence, 
$\E[\SW(\B)]
\leq 2V - F(0)\cdot V + 2v \leq 2V\left(1+ \frac{1}{\sqrt{m}}\right) -
V\left(1-\frac{1}{\sqrt{m}}\right) =V
\left(1+\frac{3}{\sqrt{m}}\right).$
On the other hand, the social welfare in the optimum
allocation is $2(V+v)=2V\left(1+\frac{1}{\sqrt{m}}\right)$ (agent $1$
is allocated $h_1$ proportion from one resource and the rest is
allocated to agent $2$). We conclude that $PoA \geq
2\frac{\left(1+\frac{1}{\sqrt{m}}\right)}{\left(1+\frac{3}{\sqrt{m}}\right)}$
which, for large $m$, converges to $2$.
\end{proof}

%%% Local Variables: 
%%% mode: latex
%%% TeX-master: "poa"
%%% End: 

%% file: budget.tex
\section{Budget Constraints}
\label{sec:budgets}
In this section, we switch to scenarios where agents have budget
constraints. We use as a benchmark the {\em effective welfare}
similarly to \cite{CV14,ST13}. We compare the effective welfare of the
allocation at equilibrium with the optimal effective welfare. We
prove an upper bound of $\phi + 1\approx 2.618$ for coarse correlated
equilibria, where $\phi = \frac{\sqrt{5}+1}{2}$ is the golden
ratio. This improves the previously known $2.78$ upper bound
in~\cite{CV14} for a single resource and concave valuations.

To prove this upper bound, we use the fact that in the equilibrium
there is no profitable unilateral deviation, and, in particular, the
utility of agent $i$ obtained by any pure deviating bid $a_i$ should be
bounded by her budget $c_i$, i.e.,
$\sum_{j \in [m]} a_{ij} \leq c_i$. We define $v^{c}$ to be the
valuation $v$ suppressed by the budget $c$, i.e.,
$v^{c}(x) = \min \{v(x),c\}$. Note that $v^c$ is also subadditive
since $v$ is subadditive. For a fixed pair $(\v,\c)$, let
$\o = (o_1,\ldots,o_n)$ be the allocation that maximizes the effective
welfare.  For a fixed agent $i$ and a vector of bids $\b_{-i}$, we
define the vector $p_i$ as $p_i = \sum_{k\neq i}b_k$. We first show
the existence of a proper deviation.

\begin{lemma}
\label{lem:deviationcoarse}
For any subadditive agent $i$, and any randomized bidding profile $\B$, there
exists a randomized bid $a_i(\B_{-i})$, such that for
any $\lambda \geq 1$, it is
  $$u_i(a_i(\B_{-i}), \B_{-i})\ge 
	\frac{ v_i^{c_i}\left( o_i \right)}{\lambda+1} -
	\frac{\sum_{j\in[m]}\sum_{k\in[n]}o_{ij}\E[b_{kj}]}{\lambda}.$$
	Moreover, for any pure strategy $\hat{a}_i$ in the support of $a_i(\B_{-i})$, 
	$\sum_j \hat{a}_{ij} \leq c_i$.
\end{lemma}
\begin{proof}
  In order to find $a_i(B_{-i})$, we define the truncated bid vector
  $\tilde{\b}_{-i}$ as follows. For any set $S\subseteq[m]$ of
  resources, we denote by $\one_S$ the indicator vector w.r.t. $S$,
  such that $x_j = 1$ for $j \in S$ and $x_j = 0$ otherwise.  For any
  vector $p_i$ and any $\lambda >0$, let $T:=T(\lambda,p_i)$ be a {\em
    maximal} subset of resources such that,
  $v_i^{c_i}(\one_T) < \frac{1}{\lambda} \sum_{j \in T} o_{ij}p_{ij}$.
  For every $k\neq i$, if $j \in T$, then $\tilde{b}_{kj} = 0$,
  otherwise $\tilde{b}_{kj} = b_{kj}$. Similarly,
  $\tilde{p}_i = \sum_{k\neq i}\tilde{b}_k$. Moreover, if
  $\b_{-i} \sim \B_{-i}$, then $p_i$ is an induced random variable
  with distribution denoted by $P_i=\{p_i|\b_{-i} \sim \B_{-i}\}$. 
	We further define distributions
  $\tilde{\B}_{-i}$ and $\tilde{P}_i$, as
  $\tilde{\B}_{-i} = \{\tilde{\b}_{-i}|\b_{-i} \sim \B_{-i}\}$ and
  $\tilde{P}_i = \{\tilde{p}_i|\tilde{\b}_{-i}\sim \tilde{\B}_{-i}\}$.

  Now consider the following bidding strategy $a_i(\B_{-i})$: sampling
  $b'_i\sim \tilde{P}_i$ and bidding
  $a_{ij}=\frac{1}{\lambda} o_{ij} b'_{ij}$ for each resource $j$. We
  first show $ \sum_{j\in[m]}a_{ij}\leq c_i$. It is sufficient to
  show that $\sum_{j\notin T} a_{ij} \leq v_i^{c_i}(\one_{[m] \setminus T})$ 
	since $v_i^{c_i}(\one_{[m] \setminus T}) \leq c_i$ and $\sum_{j\in T} a_{ij} = 0$. 
	For the sake of contradiction suppose
  $v_i^{c_i}(\one_{[m] \setminus T}) < \sum_{j
    \notin T} a_{ij}$.
  Then, by the definition of $T$ and $\tilde{p}_i$,
  $v_i^{c_i}(\one_{[m]}) \leq v_i^{c_i}(\one_{T}) +
  v_i^{c_i}(\one_{[m] \setminus T}) < \frac{1}{\lambda}\sum_{j\in T} o_{ij}p_{ij}+\sum_{j
    \notin T} a_{ij}=\frac{1}{\lambda} \sum_{j \in [m]}
  o_{ij}p_{ij}$, which contradicts the maximality of $T$.
 
  Next we show for any bid $b_i$ and $\lambda>0$,
  \begin{equation}
    \label{eq:truncated}
  v_i^{c_i}(x_i(b_i, \B_{-i})) + \frac{1}{\lambda} \sum_{j \in
    [m]}o_{ij} \E_{p_i \sim P_i}[p_{ij}] \geq v_i^{c_i}(x_i(b_i,
  \tilde{\B}_{-i})) + \frac{1}{\lambda} \sum_{j \in
    [m]}o_{ij}\E_{\tilde{p}_i \sim \tilde{P}_i}[\tilde{p}_{ij}]
  \end{equation}
  Observe that $x_i(b_i,\tilde{\b}_{-i})
  \leq x_i(b_i,\b_{-i}) +
  \one_T$.  Therefore, and by the definitions of
  $T$ and $\tilde{p}_i$, 
	\begin{eqnarray*}
	v_i^{c_i}(x_i(b_i,\tilde{\b}_{-i})) &\leq& v_i^{c_i}(x_i(b_i,\b_{-i})) + v_i^{c_i}(\one_T)\leq v_i^{c_i}(x_i(b_i,\b_{-i})) +\frac{1}{\lambda} \sum_{j \in T} o_{ij}p_{ij}\\ 
&=& v_i^{c_i}(x_i(b_i,\b_{-i})) +\frac{1}{\lambda} \sum_{j \in [m]} o_{ij}p_{ij} - \frac{1}{\lambda} \sum_{j \in [m]} o_{ij}\tilde{p}_{ij}.
\end{eqnarray*}
The claim follows by rearranging terms and taking the expectation of $\b_{-i}$, $\tilde{\b}_{-i}$, $p_i$ and 
$\tilde{p}_i$ over $\B_{-i}$, $\tilde{\B}_{-i}$, $P_i$ and $\tilde{P}_i$, 
respectively. We are now ready to prove the statement of the lemma. 
\begin{eqnarray*}
   &\E_{b'_i \sim \tilde{P_i}}&\left[u_i\left(\frac{1}{\lambda} o_i b'_i,\B_{-i}\right)\right] = \E_{b'_i \sim \tilde{P_i}}\left[v_i\left(\frac{1}{\lambda} o_i b'_i,\B_{-i}\right)\right] 
  - \frac{1}{\lambda} \sum_{j \in [m]} o_{ij} \E_{b'_i \sim \tilde{P_i}}\left[b'_{ij}\right]\\
  &\geq& \E_{b'_i \sim \tilde{P_i}}\left[v_i^{c_i}\left(\frac{1}{\lambda} o_i b'_i,\B_{-i}\right)\right] 
  - \frac{1}{\lambda} \sum_{j \in [m]} o_{ij} \E_{\tilde{p}_i \sim
    \tilde{P_i}}\left[\tilde{p}_{ij}\right] \quad\mbox{(by definition of $v_i^{c_i}$)}\\
  &\geq& \E_{b'_i \sim \tilde{P_i}}\left[v_i^{c_i}\left(\frac{1}{\lambda} o_i b'_i,\tilde{\B}_{-i}\right)\right] 
  - \frac{1}{\lambda} \sum_{j \in [m]} o_{ij} \E_{p_i \sim
         P_i}\left[p_{ij}\right] \quad\mbox{(by Inequality~\eqref{eq:truncated})}\\
  &\geq& \frac{1}{2}\E_{b'_i \sim \tilde{P_i}}\E_{\tilde{p}_i \sim \tilde{P_i}}
  \left[v^{c_i}_i\left(\frac{ o_i b'_i}{ o_i b'_i + \lambda \tilde{p}_i}+\frac{o_i \tilde{p}_i}{ o_i \tilde{p}_i + \lambda b'_i}\right)\right] 
  - \frac{1}{\lambda} \sum_{j \in [m]} o_{ij} \sum_{k\neq i} B_{kj}\\
  && \mbox{(by swapping $b'_i$ with $\tilde{p}_i$ and the subadditivity of $v_i^{c_i}(\cdot)$)}\\
  &\geq& \frac{1}{2}\E_{b'_i \sim \tilde{P_i}}\E_{\tilde{p}_i \sim \tilde{P_i}}
  \left[v^{c_i}_i\left(o_i \left(\frac{ b'_i}{b'_i + \lambda \tilde{p}_i}+\frac{\tilde{p}_i}{\tilde{p}_i + \lambda b'_i}\right)\right)\right] 
  - \frac{1}{\lambda} \sum_{j \in [m]}\sum_{k\in[n]}o_{ij}\E[b_{kj}]\\
&\geq& \frac{1}{2}v^{c_i}_i\left(\frac{2o_i}{\lambda+1} \right)
- \frac{1}{\lambda} \sum_{j \in [m]} o_{ij} \sum_{k} B_{kj} \quad \mbox{(by monotonicity of $v^{c_i}_i$)}\\
&\geq& \frac{1}{\lambda+1} v^{c_i}_i(o_i)-\frac{1}{\lambda}\sum_{j\in[m]}o_{ij}\sum_{k} B_{kj} \;\left(\mbox{subadditivity of $v^{c_i}_i$; $\frac{2}{\lambda+1} \leq 1$}\right)
\end{eqnarray*}
For the second inequality, notice that the second term doesn't depend on $b'_i$, 
so we apply Lemma 11 
for every $b'_i$. For the forth and fifth inequalities, 
$o_i \leq 1$ and $\frac{ b'_i}{ b'_i + \lambda\tilde{p}_i}+\frac{ \tilde{p}_i}{\tilde{p}_i + \lambda b'_i}
\geq \frac{2}{\lambda+1}$ for every $b'_i$, $\tilde{p}_i$ and $\lambda \geq 1$.
\end{proof}
We are ready to show the PoA bound by using the above lemma.
\begin{theorem}
The coarse correlated PoA for the proportional allocation
mechanism when agents have budget constraints and subadditive valuations, is at most $\phi +
1 \approx 2.618$. % with respect to EW
\end{theorem}
\begin{proof}
  Suppose $\B$ is a coarse correlated equilibrium. Let $A$ be the set
  of agents such that for every $i \in A$, $v_i(\B) \leq c_i$. For
  simplicity, we use $v^{c_i}_i(\B)$ to denote $\min\{\E_{\b\sim\B}[v_i(x_i(\b))],c_i\}$.
  Then for all $i \notin A$, $v^{c_i}_i(\B)=c_i\ge v^{c_i}_i(o_i)$
  and $v^{c_i}_i(\B)=c_i\ge \sum_{j \in [m]}\E[b_{ij}]$. The latter
  inequality comes from that agents do not bid higher than their
  budgets. Let $\lambda=\phi$. So $1-1/\lambda=1/(1+\lambda)$.
  By taking the linear combination and summing up
  over all agents not in $A$, we get 
\begin{equation}
\label{boundnotAcoarse}\sum_{i \notin
  A}v^{c_i}_i(\B)\geq \frac{1}{\lambda+1}\sum_{i \notin A}v^{c_i}_i(o_i) + 
\frac{1}{\lambda} \sum_{i \notin A}\sum_{j \in [m]}\E[b_{ij}]
\end{equation}
For every $i \in A$, we consider the deviating bidding strategy $a_i(\B_{-i})$ 
that is described in Lemma~\ref{lem:deviationcoarse}, then% %
\begin{eqnarray*}
v^{c_i}_i(\B) &=& v_i(x_i(\B)) = u_i(x_i(\B)) + \sum_{j \in [m]}\E[b_{ij}] \ge u_i(a_i(\B_{-i}), \B_{-i}) + \frac{1}{\lambda}\sum_{j \in [m]}\E[b_{ij}]\\
&\geq& \frac{1}{\lambda+1} v_i^{c_i}(o_i)-\frac{1}{\lambda}\sum_{j\in[m]}\sum_{k\in[n]}o_{ij}\E[b_{kj}]+\frac{1}{\lambda} \sum_{j \in [m]}\E[b_{ij}]
\end{eqnarray*}
By summing up over all $i \in A$ and by combining with
inequality~\eqref{boundnotAcoarse} we get
\begin{eqnarray*}
&&\sum_{i\in[n]}\min\{v_i(x_i(\B)),c_i\}\\
 &\geq& \frac{1}{\lambda+1} \sum_{i\in[n]}v^{c_i}_i(o_i)+\frac{1}{\lambda} \sum_{i\in[n]}\sum_{j \in [m]}\E[b_{ij}]- \frac{1}{\lambda} \sum_{i \in A}\sum_{j \in [m]} \sum_{k\in[n]}o_{ij}\E[b_{kj}]\\
&\geq & \frac{1}{\lambda+1} \sum_{i\in[n]}v^{c_i}_i(o_i)  
\qquad\qquad \mbox{$\left(\textrm{ since }\sum_{i\in A} o_{ij} \leq 1\right)$}
\end{eqnarray*}
Therefore, the PoA with respect to the effective welfare is at most $\phi + 1$. 
(recall that for Inequality~\eqref{boundnotAcoarse} we set $\lambda=\phi$)
\end{proof}

By applying Jensen's inequality for concave functions, our upper bound
also holds for the Bayesian case with single-resource and concave functions.
\begin{theorem}\label{thm:budgetConcave}
  The Bayesian PoA of single-resource proportional allocation games is at most
  $\phi + 1 \approx 2.618$, when agents have budget constraints and concave
  valuations.
\end{theorem}

\begin{proof}
  Suppose $\B$ is a Bayesian Nash equilibrium. Recall that in the
  Bayesian setting, agent $i$'s type $t_i=(v_i,c_i)$ are drawn from
  some know distribution independently. We use the notation
  $\C=(C_1,C_2,\ldots,C_n)$ to denote the bidding distribution in $\B$
  which includes the randomness of bidding strategy $\b$ and agents'
  types $\t$, that is $b_i(t_i)\sim C_i$. Then the utility of agent
  $i$ with type $t_i$ is $u_i(B_i(t_i), \C_{-i})$. Notice
  that $\C_{-i}$ does not depend on any particular $\t_{-i}$.

  Recall that $v^{c}(x) = \min \{v(x),c\}$. It is easy to check
  $v^c$ is concave if $v$ is concave. For any agents
  types $\t=(\v,\c)$, let $\o^{\t} = (o_1^{\t},...,o_n^{\t})$ be the
  allocation vector that maximizes the {\em effective welfare}.  We
  define $o_i^{t_i}$ to be the expected allocation over
  $\t_{-i}\sim\D_{-i}$ to agent $i$, in the optimum solution with
  respect to effective welfare, when her type is $t_i$. Formally,
  $o_i^{t_i} = \E_{\t_{-i} \sim \D_{-i}}[o_i^{(t_i,\t_{-i})}]$.

  For all agent $i$, let $A_i$ be the set of $t_i$ such that
  $v_i(x_i(B_i(t_i), \C_{-i})) \leq c_i$. For simplicity, we use to
  $v^{c_i}_i(B_i(t_i), \C_{-i})$ to denote
  $\min\{\E_{\t_{-i},\b\sim\B(\t)}[v_i(x_i(\b))],c_i\}$. For every
  $t_i \notin A_i$,
  $v^{c_i}_i(B_i(t_i), \C_{-i})=c_i\ge \min\{\E_{\t_{-i}}[v_i(o_i^{\t})],c_i\}$ and
  $v^{c_i}_i(B_i(t_i), \C_{-i})=c_i\ge \E[b_i(t_i)]$. The latter
  inequality comes from that agents do not bid above their budget. Let
  $\lambda=\phi$. So $1-1/\lambda=1/(1+\lambda)$. By taking the
  linear combination, taking the expectation over all $t_i \notin A_i$
  and summing up over all agents not in $A$, we get 
  \begin{equation}
    \label{boundnotAConcave}
    \sum_i\E_{t_i \notin A_i}[v^{c_i}_i(B_i(t_i), \C_{-i})]
    \geq\sum_i\E_{t_i \notin A_i}\left[\frac{1}{\lambda + 1} \min\left\{\E_{\t_{-i}}[v_i(o_i^{\t})],c_i\right\}
           + \frac{1}{\lambda}b_i(t_i)\right]
  \end{equation}

  For every $t_i \in A_i$, by Lemma~\ref{lem:deviationcoarse}, there
  exists a randomized bid $a_i(t_i,\B_{-i})$ for agent $i$, such that,
  for any $\lambda \geq 1$:
  $u_i(a_i(t_i,\B_{-i}), \B_{-i})\ge
  \frac{1}{\lambda+1}v^{c_i}_i(o_i^{t_i})- \frac{1}{\lambda}
  o_{i}^{t_i} \sum_{k\neq i} \E[b_k]$. By the definition of equilibria,
  \begin{eqnarray*}
    &&v^{c_i}_i(B_i(t_i),\C_{-i})=v_i(B_i(t_i),\C_{-i})=u_i(B_i(t_i), \C_{-i})+\E[b_i(t_i)]\\
    &\geq& u_i(a_i(t_i,\C_{-i}), \C_{-i}) + \frac{1}{\lambda}B_{i}(t_i)\geq\frac{1}{\lambda+1}v^{c_i}_i(o_i^{t_i})-
           \frac{1}{\lambda} o_{i}^{t_i}\sum_{k\in[n]}\E[b_k] + \frac{1}{\lambda} B_{i}(t_i)\\
    &\geq& \frac{1}{\lambda+1}\min\left\{\E_{\t_{-i}}[v_i(o_i^{\t})],c_i\right\}-
           \frac{1}{\lambda} o_{i}^{t_i}\sum_{k\in[n]}\E[b_k]+\frac{1}{\lambda}\E[b_i(t_i)]
  \end{eqnarray*}

  The last inequality holds due to Jensen's inequality for concave
  functions. 
  By taking the expectation over all $t_i \in A_i$, summing over all
  agents and combining with inequality \eqref{boundnotAConcave}:
  \begin{eqnarray*}
    &&\sum_i\E_{t_i}[v^{c_i}_i(B_i(t_i), \C_{-i})]\\
    &\geq& \frac{1}{\lambda + 1}\sum_i\E_{t_i}\left[\min\left\{\E_{\t_{-i}}\left[v_i\left(o_i^{\t}\right)\right],c_i\right\}\right]
    + \frac{1}{\lambda}\sum_i\E[b_i] - \frac{1}{\lambda}\sum_i\E_{t_i \in A_i}\left[ o_{i}^{t_i}\right] \sum_{k\in[n]} \E[b_k]  \\
    & \geq& \frac{1}{\lambda + 1}\sum_i\E_{t_i}\left[\min\left\{\E_{\t_{-i}}\left[v_i\left(o_i^{\t}\right)\right],c_i\right\}\right]
            + \frac{1}{\lambda}\sum_i\E[b_i]-\frac{1}{\lambda}\sum_{k\in[n]} \E[b_k]\\
    &=&\frac{1}{\lambda + 1}\sum_i\E_{t_i}\left[\min\left\{\E_{\t_{-i}}\left[v_i\left(o_i^{\t}\right)\right],c_i\right\}\right]
  \end{eqnarray*}
  The first inequality is due to that
  $\sum_i \E_{t_i}\left[ o_{i}^{t_i}\right] = \sum_i \E_{\t}\left[
    o_{i}^{\t}\right] = \E_{\t}\left[ \sum_i o_{i}^{\t}\right] \leq
  1$, since for every $\t$, $\sum_i o_{i}^{\t} \leq 1$.
  Therefore, the PoA is at most $\phi + 1$.
\end{proof}

\begin{remark}
  Syrgkanis and Tardos~\cite{ST13}, compared the social welfare in the
  equilibrium with the effective welfare in the optimum
  allocation. Caragiannis and Voudouris~\cite{CV14} also give an upper
  bound of $2$ for this ratio in the single resource case. 
	We can obtain the same upper bound by replacing $\lambda$ with $1$ 
	in Lemma~\ref{lem:deviationcoarse} and following the ideas of 
	Theorems~\ref{thm:pro_mixed} and \ref{thm:pro_bayes}.
\end{remark}

%%% Local Variables: 
%%% mode: latex
%%% TeX-master: "poa"
%%% End: 

%% file: poly.tex
\section{Polyhedral Environment}
In this section, we study the efficiency of the proportional allocation mechanism in the polyhedral
environment, that was previously studied by Nguyen and
Tardos~\cite{nguyen_approximately_2007}. We show a {\em tight} price
of anarchy bound of 2 for agents with subadditive valuations. 
Recall that, in this setting, the allocation to each agent $i$ is now
represented by a {\em single parameter} $x_i$, and not by a vector
$(x_{i1},\ldots, x_{im})$. In addition, any feasible allocation vector
$\mathbf{x}=(x_1,\dots,x_n)$ should satisfy a polyhedral constraint
$A\cdot \mathbf{x}\le \mathbf{1}$, where $A$ is a non-negative
$m\times n$ matrix and each row of $A$ corresponds to a different
resource, and $\mathbf{1}$ is a vector with all ones.  Each agent aims
to maximize her utility $u_i=v_i(x_i)-q_i$, where $v_i$ is a
subadditive 
function representing the agent's valuation. The proportional allocation mechanism determines the
following allocation and payments for each agent:
$$x_i(\b)=\min_{j:a_{ij}>0}\left\{\frac{b_{ij}}{a_{ij}\sum_{k\in
      [n]}b_{kj}}\right\};\quad\quad q_i(\b)=\sum_{j\in [m]}b_{ij},$$
where $a_{ij}$ is the $(i,j)$-th entry of matrix $A$. It is easy to verify
that the above allocation satisfies the polyhedral
constraints.  

\begin{theorem}
  \label{thm:poly_upper}
  If agents have subadditive valuations, the pure PoA of the proportional allocation mechanism in the
  polyhedral environment is exactly $2$.
\end{theorem}

\begin{proof}
  We first show that the PoA is at most $2$. Let
  $\mathbf{o}=\{o_1,\dots,o_n\}$ be the optimal allocation, $\b$ be a
  pure Nash Equilibrium, and let $p_{ij}=\sum_{k\neq i}b_{ij}$. For
  each agent $i$, consider the deviating bid $b'_i$ such that
  $b'_{ij}=o_ia_{ij}p_{ij}$ for all resources $j$. Since $\b$ is a Nash
  Equilibrium,
  \begin{eqnarray*}
    u_i(\b) &\ge& u_i(b'_i, b_{-i})=v_i\left(\min_{j:a_{ij}>0}
              \left\{\frac{o_ia_{ij}p_{ij}}{a_{ij}\left(p_{ij}+o_ia_{ij}p_{ij}\right)}\right\}\right) 
              -\sum_{j\in[m]}o_ia_{ij}p_{ij}\\
		&\ge& v_i\left(\frac{o_i}{2}\right)
		-\sum_{j\in[m]}o_ia_{ij}p_{ij}
    \ge\frac 12 v_i(o_i)-\sum_{j\in[m]}o_ia_{ij}p_{ij}
  \end{eqnarray*}
  The second inequality is true since
  $A\cdot \mathbf{x}\le \mathbf{1}$, for every allocation
  $\mathbf{x}$, and therefore $o_ia_{ij} < 1$. The last inequality
  holds due to subadditivity of $v_i$. By summing up over all agents,
  we get 
	$$\sum_iu_i(\b)\ge \frac 12 \sum_iv_i(o_i)
	-\sum_{j\in[m]}\sum_{i\in[n]}o_ia_{ij}p_{ij}
  \ge\frac 12 \sum_iv_i(o_i)
	-\sum_{j\in[m]}\sum_{k\in[n]}b_{kj}.$$
	The last inequality holds due to the fact that $p_{ij}\le \sum_{k\in[n]}b_{kj}$
  and $\sum_{i\in[n]}o_ia_{ij}\le 1$. The fact that PoA $\le 2$ follows by rearranging the
  terms.

  For the lower bound, consider a game with only two agents and a
  single resource where the polyhedral constraint is given by
  $x_1+x_2\le 1$. The valuation of the first agent is
  $v_1(x)=1+\epsilon\cdot x$, for some $\epsilon<1$ if $x<1$ and
  $v_1(x)=2$ if $x=1$. The valuation of the second agent is
  $\epsilon\cdot x$. One can verify that these two functions are
  subadditive and the optimal social welfare is $2$. Consider the
  bidding strategies $b_1=b_2=\frac {\epsilon}4 $. The utility of
  agent $1$, when she bids $x$ and agent $2$ bids $\frac {\epsilon}4$,
  is given by $1+\epsilon\cdot\frac{x}{x+\epsilon/4}-x$ which is
  maximized for $x=\frac {\epsilon}4$. The utility of agent $2$, when
  she bids $x$ and agent $1$ bids $\frac {\epsilon}4$, is
  $\epsilon\cdot\frac{x}{x+\epsilon/4}-x$ which is also maximized when
  $x=\frac {\epsilon}4$. So $(b_1,b_2)$ is a pure Nash Equilibrium
  with social welfare $1+\epsilon$. Therefore, the PoA converges to
  $2$ when $\epsilon$ goes to $0$.
\end{proof}

%%% Local Variables: 
%%% mode: latex
%%% TeX-master: "poa"
%%% End: 